\documentclass[conference,9pt,twocolumn]{IEEEtran}

\usepackage[colorlinks=true, pdfstartview=FitV, linkcolor=black, citecolor= black, urlcolor= black]{hyperref}
\usepackage{graphicx}    
\usepackage{amssymb}
\usepackage{subfig}
\usepackage[numbers]{natbib}
\usepackage{multirow}
\usepackage{bm}
\usepackage{pmat}
\usepackage{amsmath,amsthm,amsfonts,amscd} 
\usepackage{fouridx}
\usepackage{eucal} 	 	
\usepackage{verbatim}

\ifCLASSINFOpdf
\else
\fi
\newtheorem{thm}{Theorem}[section]
\newtheorem{lem}{Lemma}[section]
\newtheorem{cor}{Corollary}[section]

\theoremstyle{definition}
\newtheorem{not1}{Note}[section]
\newtheorem{PE}{Definition}[section]

\newtheorem{UCO}[PE]{Definition}

\allowdisplaybreaks[3]

\begin{document}
%
\title{Persistence based analysis of consensus protocols for dynamic graph networks}

\author{\IEEEauthorblockN{Nilanjan Roy Chowdhury}
\IEEEauthorblockA{Systems and Control Engineering\\
Indian Institute of Technology, Bombay\\
Mumbai, India 400076\\
Email: nilanjan@sc.iitb.ac.in}
\and
\IEEEauthorblockN{Srikant Sukumar}
\IEEEauthorblockA{Systems and Control Engineering \\
Indian Institute of Technology, Bombay\\
Mumbai, India 400076 \\
Email: srikant@sc.iitb.ac.in}}


%


\maketitle

\begin{abstract}
This article deals with the consensus problem involving agents with time-varying singularities in the dynamics or communication in undirected graph networks. Existing results provide control laws which guarantee asymptotic consensus. These results are based on the analysis of a system switching between piecewise constant and time-invariant dynamics. This work introduces a new analysis technique relying upon classical notions of persistence of excitation to study the convergence properties of the time-varying multi-agent dynamics. Since the individual edge weights pass through singularities and vary with time, the closed-loop dynamics consists of a non-autonomous linear system. Instead of simplifying to a piecewise continuous switched system as in literature, smooth variations in edge weights are allowed, albeit assuming an underlying persistence condition which characterizes sufficient inter-agent communication to reach consensus. The consensus task is converted to edge-agreement in order to study a stabilization problem to which classical persistence based results apply. The new technique allows precise computation of the rate of convergence to the consensus value.
\end{abstract}


%
\IEEEpeerreviewmaketitle

\section{Introduction} \label{intro}
Coordination and control of multi-agent systems is an active area of research in control theory. Control algorithms have been proposed in recent times to satisfy various cooperative control objectives, of which `consensus' forms an important subclass. Consensus is defined as a scenario in which multiple agents communicating with each other come to an agreement on a state value. The goal is to design an update law so that the vehicles in the network converge to the consensus value. Myriad applications of consensus control have been envisioned, some of which include formation control~\citep{fax2004information}, flocking~\cite{cucker2007emergent}, decentralized task assignment~\citep{alighanbari2005decentralized} and sensor networks~\citep{olfati2005consensus}.

                                                                                                                                 A multi-agent system is typically a collection of individual agents which share a common pre-specified objective. Each agent is assumed to be able to communicate with its neighbors via a network topology which usually depend on the nature of the sensor and the environment. Ideally the network topology is fixed with respect to time but most practical scenarios impose dynamically changing topologies. This may occur due to communication dropout, security reasons or intermittent actuation. 
                                                                                                                                           
                         Let us consider the modified single integrator dynamics with time varying gain~\citep{loria2005pe},
\begin{align}  
 \dot{x}_i &=g_i(t)u_i     \label{eq:tve}
   \end{align}              
where, $g_i(t)$ is a scalar time varying signal and satisfies the Persistence of Excitation (PE) condition~\citep [p. 72]{sastry2011adaptive}. The PE condition implies for scalar signals that although $g_i(t)$ might remain zero over several instants in time, there exists a window of time $ T $ over which the signal is active. Stability of dynamics in the form (\ref{eq:tve}) has been studied classically by Narendra~\citep{morgan1977stability} and in recent times by Loria et al.~\citep{loria2005pe} and Srikant and Akella~\citep{srikant2009persistence}. Narendra~\citep{morgan1977stability} has shown that the feedback law $ u = -g_i(t)x $ stabilizes the above dynamics if and only if $ g_i(t) $ is PE. In case of multi-agent systems, the time-varying scaling $g_i(t)$ can arise from the individual agent dynamics or represent the on-off nature of inter-agent communication. We consider diverse inter-agent communication topology by using $g_i(t)$ as the time-varying weight associated with edge $e_i$. The focus of this current work is to provide a Lyapunov based analysis of consensus algorithms for the aforementioned dynamics and establish a rate of convergence.\\
In~\citep{ren2007information},~\cite{ren2005survey} the authors propose a continuous-time update law for multi-agent systems. It is a well known fact that the second smallest eigenvalue of the graph laplacian matrix determines the rate of convergence in a static scenario. However in case of time-varying inter-agent weights this estimate doesn't hold true. In~\citep{ren2007information} the authors propose that the consensus control problem can also be solved by using nonlinear analysis or Lyapunov methodology. In~\citep{moreau2004stability}, the author draws attention to a non quadratic Lyapunov function to deal with the leaderless coordination problem. The focus of this article is to propose three different types of candidate Lyapunov functions for a multi-agent system communicating through a non-bidirectional and time dependent interaction topology. The authors of~\citep{slotine2005study} propose a modified partial contraction theory to resolve the group agreement and synchronization problem. In~\citep{moreau2005stability} the author imposes a condition on the communication network topology for convergence of the multi-agent system, and also proposes a new set valued Lyapunov function for the same. An amalgamation of system and graph theoretic approaches with the notion of convexity are utilized for the purpose of analysis. The conceptual framework for determining the convergence rates are further encouraged by the work of~\citep{brockett2000rate} and~\citep{ben2009stability}. In~\citep{brockett2000rate} the authors try to resolve the problem of convergence for degenerate descent procedures where as in~\citep{ben2009stability} the authors propose a new methodology for the analysis of a class of linear, degenerate gradient flow systems that is often involved in the domain of adaptive control and system identification.\\
Vicsek et al.~\citep{vicsek1995novel} formulated the consensus problem in discrete-time and proved that in absence of a central leader, all agents eventually move in the same direction. Authors of~\citep{jadbabaie2003coordination} have extended above to prove that all agents converge to a consensus value if the graph is jointly connected. In~\citep{tian2008consensus} the authors investigate the consensus problem in undirected graph networks of discrete time agents with delayed information and jointly connected topologies. Lin et al.~\citep{lin2011multi} carry out convergence analysis for single integrator with diverse time-delay, and jointly connected topologies by employing class of Lyapunov-Krasovskii functionals.\\
  A new systematic framework was proposed in~\cite{mesbahi2010graph},~\citep{zelazo2011edge} for solving the consensus problem in undirected graphs. A similarity transformation was defined in~\citep{zelazo2011edge} which relates the graph laplacian and edge laplacian matrix for any multi-agent system. The edge laplacian matrix and it's corresponding agreement protocol representation provides us a better understanding of the node agreement problem.  But the key idea defined in~\cite{mesbahi2010graph},~\citep{zelazo2011edge} is restricted to multi-agent systems with time invariant graph topology. \\
The discussion on reaching consensus of system with symmetric and diverse inter-agent communication topology can be further motivated through the work by~\citep{ren2005consensus}. The authors direct their attention to the problem of information consensus in the presence of limited information exchange and dynamically changing graph topology. Both continuous and discrete update laws are proposed to reach consensus. The key result therein shows that information consensus with minimal data availability can be achieved asymptotically if the union of the directed interaction graphs has a spanning tree. In addition the authors of~\citep{feng2011finite} extend the aforementioned idea to solve a finite time consensus problem. References~\citep{ren2007information},~\citep{feng2011finite} and~\citep{ren2005consensus} use ideas from switched systems stability and the graph laplacian to prove consensus for dynamically changing graph topologies. On the other hand, the primary focus of this paper is to utilize the edge laplacian (thus reducing the consensus problem to one of stabilization)~\citep{zelazo2011edge} and classical notions of persistence~\citep{morgan1977stability} to prove consensus for multi-agent systems with dynamic and undirected communication graphs.\\
In~\citep{arcak2007passivity} the author explicate a coordination problem and suggest a class of feedback laws that solve the above mentioned problem with local information. The author projects a passivity based approach to prove the asymptotic stability by constructing suitable Lyapunov function. It addresses bidirectional time-varying communication topology and employs the notion of persistence of excitation to prove the asymptotic stability. Though the author does not mention the convergence rate explicitly. However in this note, we are not only prove the exponential stability but deduce the convergence rate as well using the edge agreement protocol defined in~\cite{mesbahi2010graph},~\citep{zelazo2011edge}.

The article unfolds as follows. In section~\ref{Pre} a brief overview of graph theory and some preliminary ideas regarding Persistence of Excitation (PE) and stability theory are introduced. Our main result is proposed in section~\ref{mnrs} with accompanying proofs. Section~\ref{Sim} presents simulation result for a test problem. The conclusions of this work are summarized in section~\ref{conc}. 

\section{Preliminaries} \label{Pre}

\subsection{Graph Theory Terminology}
In this section some preliminary notions of graph theory are introduced. A detailed discussion of the same is given in~\citep{ren2007information},~\citep{mesbahi2010graph}. An undirected graph is a pair $(V,E)$ where 
$ V=\left\lbrace v_1,v_2,\cdots,v_n\right\rbrace  $ is a finite non empty node set and $E$ is an edge set. $ \left( i, j\right) \in E $ is an undirected edge if agents $ v_i $ and $ v_j $ exchange information with each other. In this case, edge $ v_iv_j $ is called incident with vertices $ v_i $ and $ v_j $. A path of length $ p $ in graph $ \mathcal{G} $ is given by a sequence of distinct vertices $ v_{i_0},v_{i_1},\cdots,v_{i_p} $ such that for $ k=\left\lbrace 0,1,\cdots (p-1)\right\rbrace $ the vertices $ v_{i_k} $ and $ v_{i_{ k+1 } } $ are adjacent. In this case, $ v_{i_o} $ and $ v_{i_p} $ are called end vertices of the path and $ v_{i_1},\cdots, v_{i_{p-1}} $ are called the inner vertices. When the vertices of the path are distinct except for its end vertices, the path is called a cycle. A graph is connected, if for every pair of vertices in $ V(\mathcal{G}) $, there is a path that has them as its end vertices. Any graph $ \mathcal{\tilde{G}} =\left( \tilde{V}, \tilde{E} \right) $ is a subgraph of $ \mathcal{G}=\left( V,E\right)  $ if $ \tilde{V}\subseteq V $ and $ \tilde{E}\subseteq E $. A graph without a cycle containing a single component is a tree. If for a subgraph $ V=\tilde{V} $, then it is referred to as a spanning subgraph. A spanning tree for a graph $ \mathcal{G} $ is thus a subgraph of $ \mathcal{G} $ that is also a tree. The incidence matrix $D(\mathcal G^o)$ of a graph $\mathcal G$ with node set $V=\left\lbrace v_1,v_2,....v_n\right\rbrace $, edge set $E=\left\lbrace e_1,e_2,\ldots,e_m\right\rbrace $ and arbitrary orientation 
$ \mathcal{O} $ is defined as,
\begin{align*}
D(\mathcal G^o)=\left[ d_{ij}\right]  
\end{align*}
where, \\
$\left[ d_{ij}\right]  = -1$ if $v_i$ is the tail of $e_j$\\
       $\left[ d_{ij}\right]  =  1$ if $v_i$ is the head of $e_j$\\
    $\left[ d_{ij}\right]  = 0$   otherwise\\
The arbitrary orientation doesn't effect the symmetric property of $ L(\mathcal{G}^o) $. Similarly, the graph laplacian matrix of an arbitrarily oriented graph $\mathcal G^o$ is defined as,
    \begin{align}
    L(\mathcal G^o)&=D(\mathcal G^o)D(\mathcal G^o)^T
    \end{align}
For a weighted graph the graph laplacian matrix is redefined as,
    \begin{align}
    L(\mathcal G^o)&=D(\mathcal G^o)WD(\mathcal G^o)^T
    \end{align}
    where, $W\in R^{m\times m}$ is the diagonal matrix with the weights $w(e_i), i=\left\lbrace 1,2,....m\right\rbrace $ on the diagonal entry. Similarly, the edge laplacian matrix is defined as,
    \begin{align}
    L_e(\mathcal G^o)=D(\mathcal G^o)^TD(\mathcal G^o)
    \end{align}
Since, the current article deals with undirected graphs the more cumbersome $ D(\mathcal{G}^0) $ notation will be exchanged for $ D(\mathcal{G}) $. For an undirected graph $L(\mathcal G)$ is symmetric which is not necessarily the case for a directed graph. The graph laplacian matrix is positive semi-definite with eigenvalues ordered as,
    \begin{align*}
    \lambda_1\leq\lambda_2\leq\cdots\leq\lambda_n
    \end{align*}
 where, $\lambda_1=0$. The graph $ \mathcal{G} $ is connected if and only if $\lambda_2>0$. For a directed and undirected graph $\lambda_2>0$ is defined as the algebraic connectivity and determines the convergence rate of the time invariant consensus algorithm.

\subsection{Fundamental notions of system theory}
The following classical notions of system theory are referred to, throughout the course of this paper.
\begin{PE} \label{PE1}
~\citep[p.72]{sastry2011adaptive} The signal $g(.):R\rightarrow R^{n\times m}$ is Persistently Exciting (PE) if there exist finite positive constants $\mu_1,\mu_2, T$ such that,
\begin{align}
\mu_2I_n\geq \int_t ^{t+T} g(\tau)g(\tau)^T d\tau\geq\mu_1I_n \hspace{1.5cm} \forall t\geq 0    \label{eq:pre2}
\end{align}
where, $ I_n $ is the identity matrix of dimension $ n $. In this work, whenever an assumption of persistent excitation is made on a time-dependent signal, the quantities $\mu_1$, $ \mu_2 $ and $T$ are not assumed to be explicitly known.
\end{PE}
\begin{UCO} \label{PE3}
~\citep[p.35]{sastry2011adaptive} The linear time-varying system $[A(t), C(t)] $ defined by,
\begin{align*}
\dot{x}(t) & = A(t)x(t)  \qquad x(0)=x_0 \\
y(t) & = C(t)x(t)
\end{align*}
where, $ x(t)\in R^n $ and $ y(t)\in R^m $, while $ A(t)\in R^{n\times n} $, $ C(t)\in R^{m\times n} $ are piecewise continuous functions, is called \emph{uniformly completely observable} (UCO) if there exist finite and strictly positive constants $ \beta_1, \beta_2, \delta $ such that, for all $ t_0 \geq 0 $ 
\begin{align*}
\beta_2I_n \geq \int_{t_0}^{t_0+\delta} \Phi_A^T (\tau,t_0)C^T(\tau)C(\tau)\Phi_A (\tau,t_0) d\tau \geq \beta_1I_n
\end{align*}
\end{UCO}
\begin{thm} \label{PE4}
~\citep[p.73-74]{sastry2011adaptive} Assume that, for all $ \delta>0 $, there exists $ K_{\delta} \geq 0$ such that for all, $ t_0 \geq 0 $,
\begin{align*}
\int_{t_0}^{t_0+\delta} \parallel K(\tau) \parallel^2 d\tau \leq K_{\delta} 
\end{align*} 
Then the system $ [A, C] $ is uniformly completely observable if and only if the system $ [A+KC, C] $ is uniformly completely observable. Moreover, if the observability gramian of the system $ [A, C] $ satisfies,
\begin{align*}
\beta_2I_n \geq \int_{t_0}^{t_0+\delta} \Phi_A^T (\tau,t_0)C^T(\tau)C(\tau)\Phi_A (\tau,t_0) d\tau \geq \beta_1I_n \quad \forall t_0\geq 0
\end{align*}
then the observability gramian of the system $ [A+KC, C] $ also satisfy the above mentioned inequalities with identical choice of $ \delta $ and, 
\begin{align*}
\tilde{\beta}_1 &= \frac{\beta_1}{\left( 1+ \sqrt{K_{\delta}\beta_2}\right)^2 } \\
\tilde{\beta}_2 &= \beta_2 \textit{e}^ {\left( K_{\delta} \beta_2 \right)} 
\end{align*}
\end{thm}

\begin{thm} \label{PE2}
~\citep[p.31-32]{sastry2011adaptive} Let, $ B_h $ be a closed ball of radius $ h $ centered at $ 0 $ in $ R^n $. If there exists a function $ v(t,x) $ and strictly positive constants $ \alpha_1, \alpha_2, \alpha_3, \delta $ such that for all $ x\in B_h $, $ t \geq 0 $
\begin{align*}
\alpha_1\|x\|^2\leq v(t,x)\leq\alpha_2\|x\|^2   
\end{align*}
\begin{align*}
\frac{dv(t,x (t))}{dt}\leq 0   
\end{align*}
\begin{align*}
\int_t^{t+\delta} \frac{dv(\tau, x(\tau)}{d\tau} d\tau \leq -\alpha_3\|x (t)\|^2   
\end{align*}
then, $ x(t) $ converges exponentially to $ 0 $. 
Further, $v(t,x)$ evolves according to,
\begin{align*}
v(t,x(t))\leq m_v \textit{e}^{-\alpha_v(t-t_0)} v(t_0,(x(t_0))\qquad t\geq t_0\geq0
\end{align*}
where,
\begin{align*}
m_v &=\frac{1}{\left( 1-\frac{\alpha_3}{\alpha_2}\right)} \\
\alpha_v &=\frac{1}{T} \ln \frac{1}{\left( 1-\frac{\alpha_3}{\alpha_2}\right)} 
\end{align*}
\end{thm}

\section{Main Result} \label{mnrs}
The single integrator dynamics for a class of multi-agent system is defined as~\citep{ren2008distributed}
\begin{align}
\dot{x}_i=u_i   \label{eq:mn21}
\end{align} 
with a feedback law of the following form,
\begin{align}
u_i=-k\sum_{j=1}^n a_{ij}(t)\left( x_i-x_j\right)\qquad i=1,2,\cdots,n  \label{eq:mn801}
\end{align}
Therefore, by combining the control law defined in (\ref{eq:mn801}) the overall control law depicted as,
\begin{align}
u(t) &=\begin{bmatrix}
u_1& u_2& \cdots& u_n
\end{bmatrix}^T    \nonumber  \\
&=-k L(\mathcal{G})x  \nonumber  \\
& =-kD(\mathcal G)W(t) D(\mathcal G)^T x \label{eq:mn17}
\end{align}
where, $ k>0 $, $D(\mathcal G)\in R^{ n\times m} $, is the incidence matrix for the corresponding arbitrary oriented graph and $W(t)\in R^{m\times m} $ is defined as a diagonal matrix with different $g_i^2 (t)$, $i=\left\lbrace 1,2,\ldots,m\right\rbrace $ on the diagonal entries representing the edge weights. This represents the notion of time-varying and diverse inter-agent communication topology.
Here, the diagonal matrix $W (t)\in R^{ m\times m} $ can be subdivided into two different block diagonal matrices as $ W(t)= diag \begin{bmatrix}
W_T(t),& W_C(t)
\end{bmatrix}$ 
where, $W_T (t)\in R^{p\times p}$ represents the weighting functions corresponding to the edges of the spanning tree and $W_C (t) \in R^{\left( m-p\right) \times \left( m-p\right)} $ represents the weights corresponding to the cycle edges. Here, $ p $ is the number of spanning tree edges in the given undirected graph which is equals $(n-1)$ for an undirected graph.
\begin{thm} \label{main}
The continuous update law proposed in equation (\ref{eq:mn17}) guarantees that the class of multi-agent systems with single-integrator dynamics (\ref{eq:mn21}) and time-varying communication topology characterized by $ W(t) $ achieves consensus exponentially if $ W_T(t) $ is persistently exciting. (Definition~\ref{PE1}).
\end{thm}
\begin{proof}
We begin with the interconnection between the edge laplacian and the edge agreement protocol as defined in~\citep{zelazo2011edge} in order to prove the aforementioned theorem. Consider the following similarity transformation to the edge states,
\begin{align}
x_e =D(\mathcal G)^T x \label{eq:mn5}
\end{align}
Differentiating (\ref{eq:mn5}) leads to the following expression,  
\begin{align}
\dot{x}_e &= D(\mathcal G)^T\dot{x}  \nonumber \\
 &= -kD(\mathcal G)^TD(\mathcal G)W(t)D(\mathcal G)^T x \nonumber \\
 &= - k\tilde{L}_e(\mathcal G)W(t) x_e \label{eq:mn6}
\end{align}
where the notation $ \tilde{L}_e(\mathcal{G}) $ is used to represent the time-invariant edge laplacian matrix. Now, instead of dealing with node agreement our interest shifts to the edge agreement problem since it deals with the classical stabilization problem. Let there be a connected graph $\mathcal G$ that can be described as the union of two sub-graphs as $\mathcal{G}_T \cup \mathcal{G}_C$, where $\mathcal{G}_T$ represents the spanning tree of $\mathcal G $ and $\mathcal{G}_C$ represents the cycle edges. Using an appropriate permutation of the edge indices we can partition the incidence matrix of $\mathcal G$ as,
\begin{equation}
D(\mathcal G)=\begin{bmatrix}
D(\mathcal G_T) & D(\mathcal G_C)         \label{eq:mn9}
\end{bmatrix}                  
\end{equation}
It is useful to represent the edge-laplacian matrix defined in equation (\ref{eq:mn6}) in terms of this new permutation,
\begin{align}
\tilde{L}_e(\mathcal G) &=\begin{bmatrix}
D(\mathcal G_T) & D(\mathcal G_C)
\end{bmatrix}^T \begin{bmatrix}
D(\mathcal G_T) & D(\mathcal G_C)
\end{bmatrix}        \nonumber \\
&=\begin{bmatrix}
D(\mathcal G_T)^TD(\mathcal G_T)& D(\mathcal G_T)^TD(\mathcal G_C)\\
D(\mathcal G_C)^TD(\mathcal G_T) & D(\mathcal G_C)^TD(\mathcal G_C)
\end{bmatrix}            \nonumber    \\
&=\begin{bmatrix}
\tilde{L}_e(\mathcal G_T)& D(\mathcal G_T)^TD(\mathcal G_C)\\
D(\mathcal G_C)^TD(\mathcal G_T)&\tilde{L}_e(\mathcal G_C)
\end{bmatrix}      \label{eq:mn19}
\end{align}
The matrix $W (t)\in R^{ m\times m} $ is also partitioned into two block diagonal matrices as follows,
\begin{align}
W (t)=\begin{bmatrix}
W_T (t)&0\\
0&W_C (t)
\end{bmatrix}    \label{eq:mn20}
\end{align}
where, $W_T (t)\in R^{p\times p} $ and $W_C (t) \in R^{\left( m-p\right) \times \left( m-p\right)} $ are mentioned as before. The edge state vector can be identically partitioned as,
\begin{align}
x_e=\begin{bmatrix}
x_T\\
x_C
\end{bmatrix}     \label{eq:mn10}
\end{align}
The columns of the cycle edges $D(\mathcal G_C)$ are linearly dependent on the columns of $D(\mathcal G_T)$. This relationship can be expressed as follows,
\begin{align}
D(\mathcal G_T)T=D(\mathcal G_C)                                  
\end{align}
where, the matrix $ T $ is defined as,
\begin{equation}
T=\left( D(\mathcal G_T)^T D(\mathcal G_T)\right) ^{-1} D(\mathcal G_T)^TD(\mathcal G_C)
\end{equation}     
Substituting (\ref{eq:mn19}) and (\ref{eq:mn20}) in (\ref{eq:mn6}) we have,
\begin{align}
\dot{x}_e &=-k\begin{bmatrix}
\tilde{L}_e(\mathcal G_T)& D(\mathcal G_T)^TD(\mathcal G_C)\\
D(\mathcal G_C)^TD(\mathcal G_T)&\tilde{L}_e(\mathcal G_C)
\end{bmatrix} \begin{bmatrix}
W_T (t)&0\\
0&W_C(t)
\end{bmatrix} x_e 
\end{align}
Therefore, the states corresponding to the spanning tree and cycle edges evolve according to,
\begin{align}
\dot{x}_T &=-k\tilde{L}_e(\mathcal G_T)W_T (t)x_T- kD(\mathcal G_T)^TD(\mathcal G_C)W_C (t)x_C  \\ \label{eq:mn7}
\dot{x}_C &=-kD(\mathcal G_C)^TD(\mathcal G_T)W_T (t)x_T-k\tilde{L}_e(\mathcal G_C)W_C (t)x_C
\end{align}
Here, we are strictly interested in the behavior of the edges corresponding to the spanning tree, since they represent the minimal edge subset that must go to zero for consensus to be achieved. The cycle edges can be reconstructed from the spanning tree edges as follows,
\begin{align}
x_C(t)=T^Tx_T(t)
\end{align}
With the aforementioned transformation, the $ x_T $ dynamics reduces to,
\begin{align}
\dot{x}_T &=-k\tilde{L}_e(\mathcal G_T)W_T(t)x_T- kD(\mathcal G_T)^TD(\mathcal G_C)W_C(t)x_C  \nonumber \\
& = -k\tilde{L}_e(\mathcal G_T)\left[ W_T(t)+TW_C(t)T^T\right] x_T     \nonumber \\
& =-k\tilde{L}_e(\mathcal G_T)\begin{bmatrix} 
I & T
\end{bmatrix} \begin{bmatrix}
W_T(t)&0\\
0&W_C(t)
\end{bmatrix} \begin{bmatrix}
I\\
T^T
\end{bmatrix}     
x_T               \nonumber \\
& = -k\tilde{L}_e(\mathcal G_T) RW(t)R^Tx_T    \label{eq:mn8}
\end{align}
where, $ R=\begin{bmatrix}
I & T
\end{bmatrix} $. It is evident that $\tilde{L}_e(\mathcal{G}_T) $ is symmetric and positive definite as long as there exists a spanning tree. It can therefore be decomposed as 
$\tilde{L}_e(\mathcal G_T)=\Gamma \Lambda \Gamma^T$, substituting into the spanning tree equation described in (\ref{eq:mn8}) yields,
\begin{equation}
\dot{x}_T=-k\Gamma \Lambda \Gamma^TRW(t)R^Tx_T
\end{equation}
We now introduce a set of modified states defined by the similarity transformation 
$y=\Gamma^Tx_T $ with dynamics,
\begin{align}
\dot{y} = \Gamma^T\dot{x}_T = - k\Lambda \Gamma^T RW(t)R^T \Gamma y \label{eq:mn12}     
\end{align}
Therefore, we need to establish the exponential convergence of (\ref{eq:mn12}) instead of (\ref{eq:mn8}). In the preceding analysis, $\Lambda$ is a diagonal matrix with positive entries. It contains the nonzero eigenvalues of the edge laplacian matrix associated with the spanning tree. 

For the proof of exponential convergence we define a Lyapunov-like candidate function,
\begin{align}
& V\left(y\right) =y^T \Lambda^{-1} y  
\end{align}
Where, $ \Lambda^{-1} $ is a diagonal matrix with positive eigenvalues. The time derivative of $ V\left( y\right)  $ along closed-loop dynamics~(\ref{eq:mn12}) can be written as,
\begin{align}
\Rightarrow \dot{V}\left(y\right) & = y^T\Lambda^{-1}\dot{y}+\dot{y}^T\Lambda^{-1}y \nonumber \\
& = y^T \Lambda^{-1} \left(-k\Lambda\Gamma^TRW(t)R^T\Gamma\right) y \nonumber  \\ & + y^T \left(-k\Gamma^TRW(t)R^T\Gamma\Lambda^T\right)\Lambda^{-1} y   \nonumber \\
& = -2ky^T \left(\Gamma^TRW(t)R^T\Gamma\right) y 
\end{align}
Integrating both sides we get,
\begin{align}
& \int_t^{t+T} \dot{V}\left(\sigma,y\right) d\sigma = -2k\int_t^{t+T} y^T(\sigma)\Gamma^TRW(\sigma)R^T\Gamma  y(\sigma) d\sigma \nonumber \\
&= - 2k\int_t^{t+T} y^T(\sigma)\left( \Gamma^TRW(\sigma)^{\frac{1}{2}}\right) \left( W(\sigma)^{\frac{1}{2}} R^T \Gamma\right) y(\sigma) d\sigma   \label{eq: mn22}
\end{align}

\begin{lem}
$W(t)^{1/2} $ is PE $\Rightarrow$ $ RW(t)^{1/2} $ is PE $\Rightarrow$ $ \Gamma^T RW(t)^{1/2} $ is PE.
\end{lem}
\begin{proof}
We begin by establishing persistence of $RW(t)^{1/2}$.
\begin{align}
\int_t^{t+T}RW(\sigma)R^T d\sigma &= \int_t^{t+T}\left[ W_T(\sigma)+TW_C(\sigma)T^T\right] d\sigma  \label{eq:mn30}
\end{align}
Consider any vector $\alpha$ such that $\|\alpha\|=1$,

\begin{align*}
& \alpha^T \left\lbrace \int_t^{t+T}\left[ W_T(\sigma)+TW_C(\sigma)T^T\right] d\sigma\right\rbrace \alpha \nonumber \\
& = \alpha^T diag \left[ \int_t^{t+T}g_1^2(\sigma),\cdots,\int_t^{t+T}g_p^2(\sigma) d\sigma\right] \alpha  \\ & +\left( T^T \alpha\right)^T diag \left[ \int_t^{t+T}g_1^2(\sigma),\cdots,\int_t^{t+T}g_{m-p}^2(\sigma) d\sigma\right] \left( T^T \alpha\right)   \\
& \geq \mu \alpha^T \alpha + \mu \left( T^T \alpha\right)^T \left( T^T \alpha\right) 
\end{align*}
where we have utilized persistence of the signal $W(t)$ to arrive at lower bounds on the quadratic form.
\begin{not1}
In the above inequality $\mu >0$ and $ \|\alpha\|=1 $ . Therefore it can be easily concluded that $\mu \alpha^T \alpha = \mu$. On the other hand, $\mu \left( T^T \alpha\right)^T \left( T^T \alpha\right) \geq 0  $.
\end{not1}
Hence from the above mentioned analysis the conclusion can be rephrased as,
\begin{align}
\alpha ^T \left\lbrace \int_t^{t+T}  RW(\sigma)R^T d\sigma\right\rbrace \alpha  &\geq \mu \alpha^T \alpha + \mu \left( T^T \alpha\right)^T \left( T^T \alpha\right)   
\geq \mu
\end{align}
It is also evident using similar arguments, that a similarity transform does not impact the persistence of a signal and therefore $\Gamma^T R W(t) R^T \Gamma$ is also persistently exciting with identical $T$ and $\mu$.
\end{proof}

Therefore, from the PE condition defined in definition~\ref{PE1} and definition~\ref{PE3}, it is evident that $\left[ 0,W(t)^{\frac{1}{2}} R^T \Gamma \right] $ is UCO. 
Define $ K(t)=-k\Lambda\Gamma^TRW^\frac{1}{2}(t) $. The integral of $ K(t) $ over a window of time $ T $ can be evaluated as,
\begin{align*}
\int_{t_0}^{{t_0}+T} \parallel K(\sigma) \parallel^2 d\sigma &= \int_{t_0}^{{t_0}+T} k\parallel \Lambda \Gamma^T RW(\sigma)^{\frac{1}{2}} d\sigma \parallel^2   \\
& \leq k\parallel \Lambda \parallel^2 \left\lbrace tr \int_{t_0}^{{t_0}+T} \Gamma^T RW(\sigma) R^T \Gamma d\sigma\right\rbrace   \\
& \leq k\parallel\Lambda\parallel^2 p \mu_2
\end{align*}
where, $ p $ is the dimension of $\Gamma^T RW(t) R^T \Gamma $ denotes the number of spanning tree edges. Then, by theorem~\ref{PE4} the system $\left[ 0,W(t)^{\frac{1}{2}} R^T \Gamma \right] $ is UCO if and only if $\left[ -k\Lambda (\Gamma^TRW(t)^{\frac{1}{2}}) (W(t)^{\frac{1}{2}} R^T \Gamma), W(t)^{\frac{1}{2}} R^T \Gamma\right] $ is UCO. The observability gramian for this modified system with, $ A(t) =-k\Lambda\Gamma^TRW(t)R^T\Gamma $ is as follows,
\begin{align}
\tilde{\mu}_2 I_n \geq \int_{t_0}^{t_0+\delta} \Phi_A^T (\tau,t_0)C^T(\tau)C(\tau)\Phi_A (\tau,t_0) d\tau \geq \tilde{\mu}_1 I_n  \label{eq:mn25}
\end{align}
where,
\begin{align*}
\tilde{\mu}_1 & =\frac{\mu_1}{\left( 1+k\sqrt{p} \parallel\Lambda\parallel \mu_2 \right)^2}   \\
\tilde{\mu}_2 & = \mu_2 \textit{e}^ {k p \parallel\Lambda \parallel^2 \mu_2^2} 
\end{align*}
Therefore, for all $ t\geq t_0 $, the integral defined in equation (\ref{eq: mn22}) evaluates to,
\begin{align} 
& \int_{t_0}^{{t_0} +T} \dot{V}(\sigma) d\sigma = \nonumber \\
& -2ky^T(0) \left\lbrace  \int_{t_0}^{{t_0} +T} \Phi_A^T (\sigma,t_0)\Gamma^TRW(\sigma)R^T\Gamma\Phi_A (\sigma,t_0) d\sigma \right\rbrace  y(0)  
\end{align}
which by applying the modified UCO condition~\ref{PE4} yields,
\begin{align}
\Rightarrow \int_{t_0}^{{t_0}+T} \dot{V}(\sigma) d\sigma & \leq -\frac{2k \lambda_{min} (\Lambda) \mu_1}{\left( 1+k\sqrt{p} \parallel \Lambda\parallel\mu_2 \right)^2} V(y_0) 
     \label{eq: mn24}
\end{align}
Now comparing the result in equation (\ref{eq: mn24}) to the exponential stability theorem~\ref{PE2}, we have,
\begin{align*}
\frac{\alpha_3}{\alpha_2} &= \frac{2k \lambda_{min} (\Lambda) \mu_1}{\left( 1+k\sqrt{p} \parallel \Lambda \parallel \mu_2 \right)^2}    \\
m_v &=\frac{1}{\left( 1-\frac{\alpha_3}{\alpha_2}\right)}= \frac{1}{\left[ 1-\frac{2 k\lambda_{min} (\Lambda) \mu_1}{\left( 1+k\sqrt{p} \parallel \Lambda \parallel \mu_2 \right)^2}\right] }  \\
\alpha_v &=\frac{1}{T} \ln \frac{1}{\left( 1-\frac{\alpha_3}{\alpha_2}\right)}=\frac{1}{T} \ln \frac{1}{\left[ 1-\frac{2k \lambda_{min} (\Lambda) \mu_1}{\left( 1+k\sqrt{p} \parallel \Lambda \parallel \mu_2 \right)^2}\right] }  
\end{align*} 
and the explicit solution can be rephrased as,
\begin{align}
V(y(t)) & \leq m_v \textit{e}^{-\alpha_v(t-t_0)} V(y_0)  \nonumber \\
\Rightarrow \parallel y(t)\parallel & \leq \left\lbrace \sqrt{\frac{\lambda_{max} (\Lambda) m_v}{\lambda_{min} (\Lambda)}}\right\rbrace  \textit{e}^{-\frac{\alpha_v}{2} (t-t_0)}\parallel y(0)\parallel
\end{align}

Hence the convergence rate is determined as follows,
\[
\boxed{\dfrac{\alpha_v}{2} ={\dfrac{1}{2T} \ln \frac{1}{\left[ 1-\frac{2k \lambda_{min} (\Lambda) \mu_1}{\left( 1+k\sqrt{p} \parallel \Lambda \parallel \mu_2 \right)^2}\right]  }}}
\]
\begin{not1}
Here $ y(t) $ and $ x_T(t) $ are related to each other via the similarity transformation, $ y(t)=\Gamma^Tx_T(t) $ where, $ \Gamma $ is full column rank. It is a well known fact that the convergence rate doesn't change under a linear transformation. Hence, the convergence rate remains unaltered.
\end{not1}
\end{proof} 
We wish to influence the rate of convergence solely by changing the value of gain $  k$. For a time invariant graph network, arbitrarily pushing up the value of $ k $ improves the rate of convergence. However, this does not hold true in the time-varying scenario. In the time-varying case, even though $ k $ can be used to increase the convergence rate to an extent, arbitrarily increasing $ k $ does not guarantee large rates of exponential convergence. This is evident from the expression of the convergence rate, which has $ k $ in both the numerator and denominator. The effect of this scalar gain on the convergence rate will be clarified through simulations later.

\begin{cor}
~\citep [p.45-46] {ren2008distributed} Let $ t_1, t_2,\cdots $ be the infinite time sequence such that $ \tau_i=t_{i+1}-t_i $, $ i=0,1, \cdots $. Let $\mathcal{G}_n(t_i) $ be the undirected graph at time $ t=t_i $ with non-negative edge weights. Continuous time algorithm (\ref{eq:mn17}) achieves consensus asymptotically if there exists an infinite sequence of contiguous, nonempty, uniform bounded, time intervals $\left[ t_{ij}, t_{i j+1 }\right)$; $ j=1,2,\cdots $ starting at $ t_{i1}=t_0 $ with the property that the union of the undirected graph across each such interval has the same spanning tree. 
\end{cor}
Firstly, we can assume that all graphs are complete (i.e. all nodes are in the neighborhood of the other). If not then the weights corresponding to those edges are permanently set to zero. This ensures that at each instant in time the number of edge states stay constant. Now, it only needs to be proven that the conditions in the corollary imply that $W_T(t)$ is persistently exciting. We have assumed that the all intervals are uniformly bounded. Let this be uniform bound be $t_{max}$. Further, union of graphs in each contiguous interval contain the same spanning tree. This implies that the spanning tree edge states do not change between intervals. This is sufficient because we are only concerned with the convergence of the spanning edge states.\\
We now choose $T > 2t_{max}$ in Definition~\ref{PE1}. If we integrate $W_T^2(t)$ over any window of time $T$ we can show that it satisfies a positive definite lower bound as required by (\ref{eq:pre2}). This proves PE of $W_T(t)$ and Theorem~\ref{main} can be directly applied to prove convergence.

\section{Simulation} \label{Sim}
In this section, we consider an example to validate the result in theorem~\ref{main}. Here we consider a multi-agent system with four agents. The dynamics of the multi-agent system are defined as,
\begin{align*}
\dot{x}_i=u_i\\
\dot{y}_i=v_i
\end{align*}
where $ \left( x_i, y_i\right)  $ denote the coordinates of the $ i^{th} $ agent. The graphical representation (with arbitrary orientation) for the above mentioned multi-agent system is shown in Fig~\ref{draw} with $ g_i^2(t) $ representing the weights corresponding to edge $ e_i $. The incidence matrix defined in section~\ref{Pre} is calculated as, 
\begin{figure}[!htb]
	\centering
		\includegraphics[width = 2.5in, keepaspectratio]{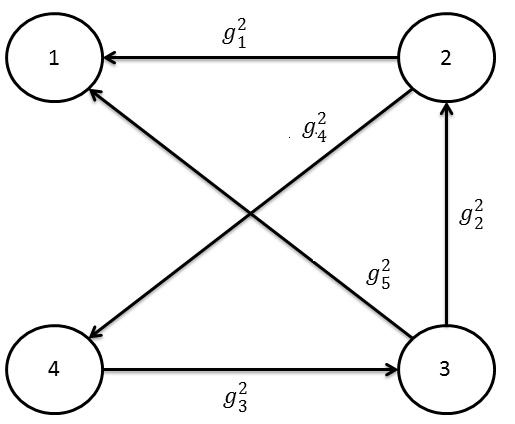}
	\caption{ Information-exchange topologies between the four agents}\label{draw}
\end{figure}
\begin{align*}
D(\mathcal{G})=\begin{bmatrix}
1&0&0&0&1\\
-1&1&0&-1&0\\
0&-1&1&0&-1\\
0&0&-1&1&0
\end{bmatrix}
\end{align*} 
The weight matrix is defined as $ W(t)=\text{diag}[g_1^2(t),g_2^2(t),g_3^2(t),g_4^2(t),g_5^2(t)] $ where $ g_i^2(t)=\left\lbrace \text{square}(t)+1\right\rbrace^2 \text{sin}^2(it) $ for $ i=\left\lbrace 1,2,3,4,5\right\rbrace$ with duty cycles of $ 0.1, 0.2, 0.3, 0.4, 0.5 $ respectively. The weight matrix represents the dynamic changes in the inter agent communication graph. The individual weight, $ g_i^2(t) $ signify the on-off phases in the communication. The initial conditions are chosen as, $ x_0=[0.1,0.1,0.77,0.8]^T $ and $ y_0=[0.8,0.1,0.78,0]^T $. The results obtained from simulations is shown in Fig.~\ref{consen} and Fig.~\ref{conv}. The control law defined in (\ref{eq:mn17}) with control gain $ k=1 $ directs the four agents to move from their initial locations smoothly to the consensus value shown in the inset of Fig.~\ref{consen}. Fig.~\ref{conv} plots the convergence of the spanning edge states along with the exponential envelope. As expected, the evolution of the spanning edge states always lies within the estimated exponential convergence rate envelope.

 \begin{figure}[!htb]
	\centering
		\includegraphics[width = 3.5 in, keepaspectratio]{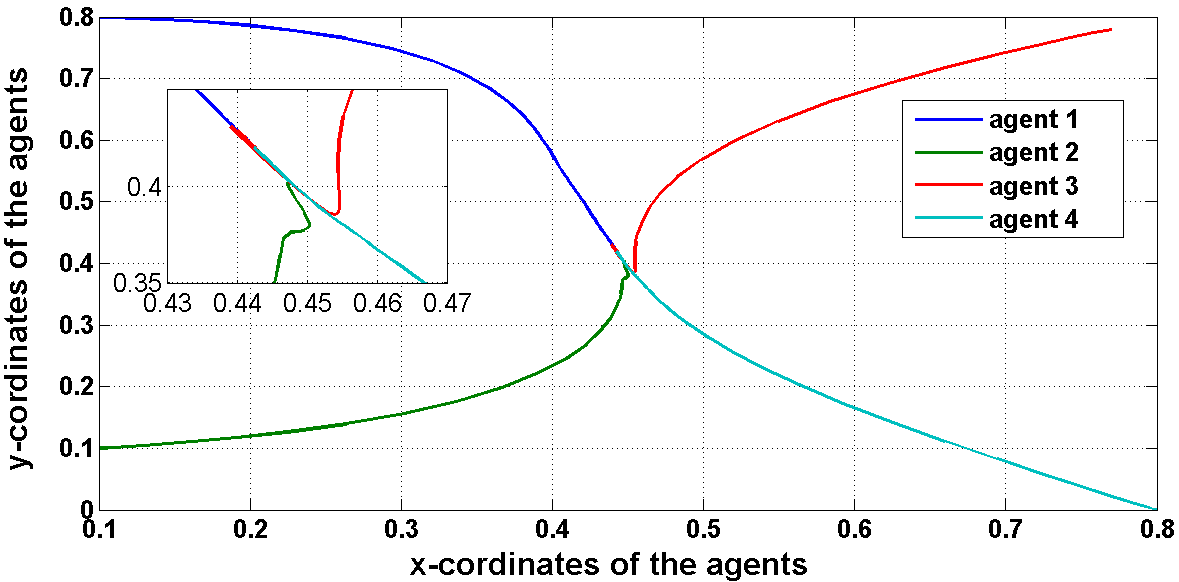}
	\caption{ Resulting state trajectories of the four agents}\label{consen}
\end{figure}

\begin{figure}[!htb]
	\centering
		\includegraphics[width = 3.5 in, keepaspectratio]{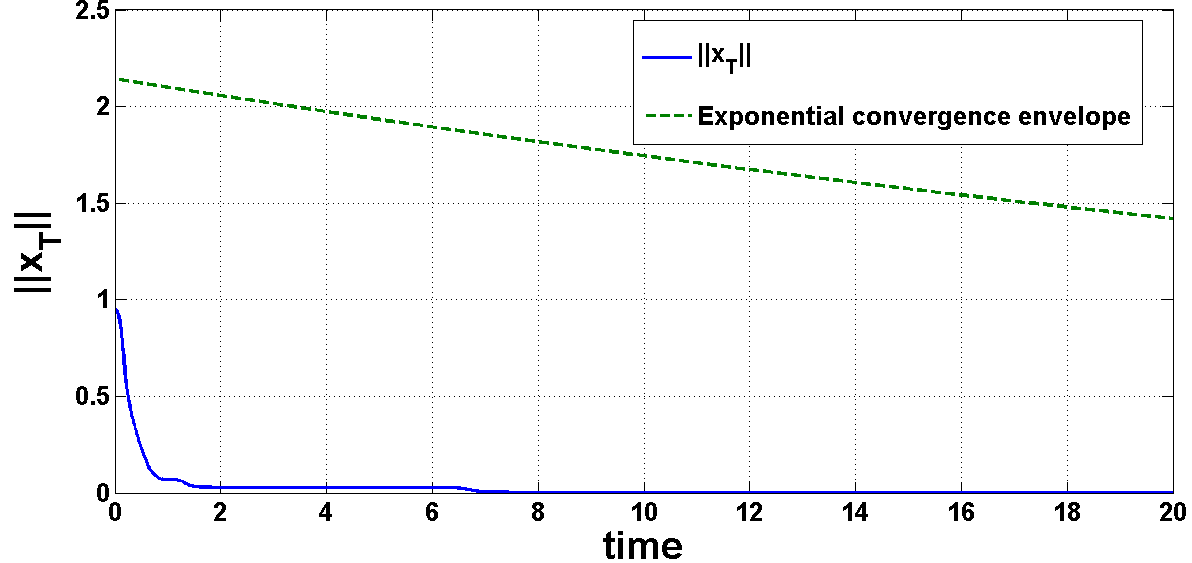}
	\caption{ Convergence of $ ||x_T|| $ with exponential convergence envelope}\label{conv}
\end{figure}
Similarly, we consider the control law defined in (\ref{eq:mn17}) with the aforementioned agent dynamics. The simulation results for the spanning edge vectors are shown in Fig.~\ref{kconv}, with different control gains.
\begin{figure}[!htb]
	\centering
		\includegraphics[width = 3.5 in, keepaspectratio]{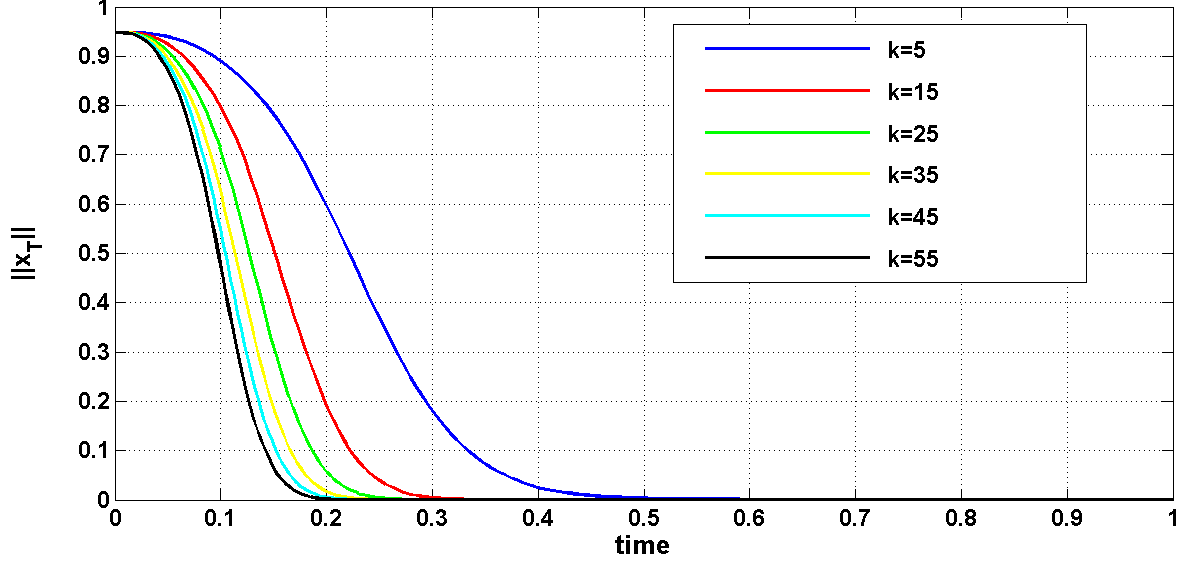}
	\caption{ Convergence of $ ||x_T|| $ with different values of $ k $}\label{kconv}
\end{figure}
The control law defined in (\ref{eq:mn17}) with different control gain allows the spanning edge vectors to converge exponentially to the origin from their arbitrary non-zero initial condition. From the simulation result it can be observed that, the exponential convergence rate increases directly with the scalar gain, for a lower range of values of the design parameter $k$. However this rate saturates for higher values of the gain, as expected from our theoretical analysis.
\section{Concluding Remarks} \label{conc} 
A uniform consensus algorithm for the class of multi-agent systems communicating through diverse inter-agent communication topologies is studied in this work. Time-varying weights are assigned corresponding to each edge which potentially pass through singular phases representing communication dropouts. However, these time-varying weights are assumed to satisfy a persistence of excitation condition. The consensus control law is analyzed by transforming the node agreement problem to an edge agreement one by a suitable coordinate transformation. This allows us to look at a stabilization problem, thus allowing utilization of classical results in adaptive control to prove consensus. The time-dependent control scheme is shown to exponentially stabilize the edge set vector for the multi-agent system with dynamic communication topology. The new analysis technique employing notions of persistence of excitation also helps compute the exponential rate of convergence for the closed-loop dynamics. A modified control law introducing a constant, scalar gain is also studied with the aim of improving convergence rates to consensus. It is observed that, though convergence rates improve for small increase in the gain, large increments in the gain do not arbitrarily push up the rate of convergence to consensus value. Therefore, future work involves exploring whether the rate of convergence to the consensus value can be improved by modifying the control law.
\\

\bibliographystyle{IEEEtranN}
\bibliography{diss}
%

\end{document}